\documentclass[10pt,twocolumn,journal]{IEEEtran}

\usepackage{amsmath, mathrsfs}
\usepackage{amsfonts}
\usepackage{amssymb}
\usepackage{graphicx}
\usepackage{color}
\usepackage{multirow}
\usepackage{graphicx}
\usepackage{stmaryrd}
\usepackage[none]{hyphenat}
\usepackage{mathtools}

\allowdisplaybreaks

\usepackage{amsfonts}
\usepackage{amssymb}
\usepackage{graphicx}
\usepackage{color}
\usepackage{multirow}
\usepackage{graphicx}

\usepackage{caption}
\usepackage{subcaption}
\usepackage[numbers,sort&compress]{natbib}

\newcommand{\subparagraph}{}
\usepackage[compact]{titlesec}
\titlespacing*{\section}{15pt}{1\baselineskip}{0.9\baselineskip}

\allowdisplaybreaks

\newcommand{\myhash}{%
  {\settoheight{\dimen0}{C}\kern-.05em\, \resizebox{!}{\dimen0}{\raisebox{\depth}{\#}}}}

\usepackage{multicol}


\usepackage{mathbbol}

\def\mindex#1{\index{#1}}

\def\sq{\hbox{\rlap{$\sqcap$}$\sqcup$}}
\def\qed{\ifmmode\sq\else{\unskip\nobreak\hfil
\penalty50\hskip1em\null\nobreak\hfil\sq
\parfillskip=0pt\finalhyphendemerits=0\endgraf}\fi\medskip}

\long\def\defbox#1{\framebox[.9\hsize][c]{\parbox{.85\hsize}{%
\parindent=0pt
\baselineskip=12pt plus .1pt      
\parskip=6pt plus 1.5pt minus 1pt 
 #1}}}

\long\def\beginbox#1\endbox{\subsection*{}%
\hbox{\hspace{.05\hsize}\defbox{\medskip#1\bigskip}}%
\subsection*{}}

\def\endbox{}

\newsavebox{\junk}
\savebox{\junk}[1.6mm]{\hbox{$|\!|\!|$}}

\def\liminf{\mathop{\rm lim\ inf}}

\def\bC{{\mathbb C}}

\def\bE{{\mathbb E}}

\def\bH{{\mathbb H}}

\def\bZ{{\mathbb Z}}

\def\bbh{{\mathbb h}}

\def\bfA{{\bf A}}

\def\bfH{{\bf H}}
\def\bfI{{\bf I}}

\def\bfa{{\bf a}}

\def\bfc{{\bf c}}

\def\bfg{{\bf g}}
\def\bfh{{\bf h}}

\def\bfq{{\bf q}}

\def\bfu{{\bf u}}

\def\bfx{{\bf x}}
\def\bfy{{\bf y}}
\def\bfz{{\bf z}}

\def\sfH{{\sf H}}

\def\bfmath#1{{\mathchoice{\mbox{\boldmath$#1$}}%
{\mbox{\boldmath$#1$}}%
{\mbox{\boldmath$\scriptstyle#1$}}%
{\mbox{\boldmath$\scriptscriptstyle#1$}}}}

\def\bfmY{\bfmath{Y}}

\def\bfmhhaY{\bfmath{\hhaY}} 
\def\bfmhhaY{\hbox to 0pt{$\widehat{\bfmY}$\hss}\widehat{\phantom{\raise 1.25pt\hbox{$\bfmY$}}}}

\def\til={{\widetilde =}}

 \def\FRAC#1#2#3{\genfrac{}{}{}{#1}{#2}{#3}}

\def\ddtp{{\mathchoice{\FRAC{1}{d^{\hbox to 2pt{\rm\tiny +\hss}}}{dt}}%
{\FRAC{1}{d^{\hbox to 2pt{\rm\tiny +\hss}}}{dt}}%
{\FRAC{3}{d^{\hbox to 2pt{\rm\tiny +\hss}}}{dt}}%
{\FRAC{3}{d^{\hbox to 2pt{\rm\tiny +\hss}}}{dt}}}}

\def\average#1,#2,{{1\over #2} \sum_{#1}^{#2}}

\def\eye(#1){{\bf(#1)}\quad}

\newtheorem{theorem}{{\bf Theorem}}

\def\eq#1/{(\ref{e:#1})}

\newcommand{\inp}[2]{{\langle #1, #2 \rangle}}

\newcommand{\beqn}[1]{\notes{#1}%
\begin{eqnarray} \elabel{#1}}

\newcommand{\eeqn}{\end{eqnarray} }

\newcommand{\beq}[1]{\notes{#1}%
\begin{equation}\elabel{#1}}

\newcommand{\eeq}{\end{equation}}

\def\bdes{\begin{description}}
\def\edes{\end{description}}

\newcounter{rmnum}

\newcounter{anum}

{\end{list}}

\def\ass(#1:#2){(#1\ref{#1:#2})}

\def\ritem#1{
\item[{\sf \ass(\current_model:#1)}]
}

\newenvironment{recall-ass}[1]{%
\begin{description}
\def\current_model{#1}}{
\end{description}
}

\usepackage{tikz}
\usepackage{pgfplots,tikz-3dplot}
\pgfplotsset{compat=newest}
\usetikzlibrary{patterns}
\usetikzlibrary{positioning}
\usetikzlibrary{datavisualization}
\usetikzlibrary{datavisualization.formats.functions}
\usetikzlibrary{backgrounds}
\usetikzlibrary{shapes,snakes}

\usepgfplotslibrary{external} 
\usepackage{float}
\usepackage{afterpage}
\usepackage{balance}

\def\herm{{\sfH}}

\long\def\comment#1{}

\newcommand{\Sigmam}{\hbox{\boldmath$\Sigma$}}

\renewcommand{\arg}{{\hbox{arg}}}

\newcommand{\transp}{{\sf T}}

\author{ Mahdi Barzegar Khalilsarai$^*$, Saeid Haghighatshoar$^*$, Giuseppe Caire$^*$, Gerhard Wunder$^\dagger$\\
	\vspace{2mm}
	
	$^*\{$m.barzegarkhalilsarai, saeid.haghighatshoar, caire$\}$@tu-berlin.de, $^\dagger$wunder@zedat.fu-berlin.de    
	\vspace{-3mm}
}

\graphicspath{ {./Figures/} }

\begin{document}

\title{Compressive Estimation of a Stochastic Process with Unknown Autocorrelation Function}

\maketitle

\begin{abstract}
	In this paper, we study the prediction of a circularly symmetric zero-mean stationary Gaussian process from a window of observations consisting of finitely many samples. This is a prevalent problem in a wide range of applications in communication theory and signal processing. Due to stationarity, when the autocorrelation function or equivalently the power spectral density (PSD) of the process is available, the \textit{Minimum Mean Squared Error} (MMSE) predictor is readily obtained. In particular, it is given by a linear operator that depends on autocorrelation of the process as well as the noise power in the observed samples. The prediction becomes, however,  quite challenging when the PSD of the process is unknown. In this paper, we propose a \textit{blind} predictor that does not require the a priori knowledge of the PSD of the process and compare its performance with that of an MMSE predictor that has a full knowledge of the PSD. To design such a blind predictor, we use the random spectral representation of a stationary Gaussian process. We apply the well-known atomic-norm minimization technique to the observed samples to obtain a discrete quantization of the underlying random spectrum, which we use to predict the process. Our simulation results show that this estimator has a good performance comparable with that of the MMSE estimator.
\end{abstract}


\section{Introduction}
Let $\bbh:=\{h_n: n\in \bZ_+\}$ be a zero-mean circularly symmetric stationary Gaussian process. Due to stationarity, the probability law of the process is fully characterized by its autocorrelation function $r_h(m)=\bE[h_nh^*_{n-m}]$. Let $N\in \bZ_+$ and let 
\begin{align}
y_n = h_n + z_n,\ \ n\in [N]
\end{align}
be a window of observations of the process consisting of $N$ samples, where $[N]=\{0,1,\dots, N-1\}$ and where $\{z_n: n \in [N]\}$ denotes the observation noise consisting of $N$ i.i.d. zero-mean circularly symmetric Gaussian variables with a variance $\sigma^2$. 
In this paper, we are interested in predicting the process $\bbh$ over the window $\Omega=\{N,N+1, \dots, 2N-1\}$ consisting of $N$ future samples of $\bbh$ using the observed samples $y_n, n\in [N]$. Denoting $\bfy=(y_0, \dots, y_{N-1})^\transp$, $\bfh=(h_0,\dots, h_{N-1})^\transp$, $\bfz=(z_0, \dots, z_{N-1})^\transp$, and $\bfg=(h_N, \dots, h_{2N-1})$, the optimal \textit{Minimum Mean Squared Error} (MMSE) predictor of $\bfg$ from the observations $\bfy$ is given by the linear operator
\begin{equation}\label{MMSE_estimate}
\widehat{\bfg}=\Sigmam_{\bfg\bfy} \Sigmam_{\bfy\bfy}^{-1} \bfy=\Sigmam_{\bfg\bfh} \Sigmam_{\bfy\bfy}^{-1} \bfy,
\end{equation}
where $\Sigmam_{\bfg\bfy}:=\bE[\bfg \bfy^\herm]=\Sigmam_{\bfg\bfh}$ and $\Sigmam_{\bfy\bfy}=\bE[\bfy\bfy^\herm]=\Sigmam_{\bfh\bfh} + \sigma^2 \bfI_N$ denote the cross-correlation matrix of $\bfg$ and $\bfy$ and the autocorrelation matrix of $\bfy$ respectively, and where we used the independence of $\bfh$ and the observation noise $\bfz$. Notice that the linear MMSE predictor only depends on $\Sigmam_{\bfg \bfh}$ and $\Sigmam_{\bfh\bfh}$ whose components can be obtained from the autocorrelation function $r_h(m)$ of the process $\bbh$. Thus, when the autocorrelation function  or equivalently the \textit{Power Spectral Density} (PSD) of the process is known a priori, the predictor in \eqref{MMSE_estimate} can be directly computed.  

The prediction becomes quite challenging, however, if the autocorrelation function is unknown, which is the case in many applications in communication theory and signal processing. For instance, this problem arises in a wireless communication scenario where one needs to obtain the channel state information, modeled by a Gaussian process, to schedule a set of users, which in turn requires predicting a fading channel in time. The multipath fading is commonly modeled as a superposition of sinusoids with random frequencies, amplitudes and phases \cite{gallager2008principles}. In \cite{shirani2010mimo}, the authors exploit such a discrete model for the fading channel to estimate the corresponding parameters via ESPRIT, as exploited for \textit{Direction of Arrival} (DOA) estimation algorithm in \cite{roy1989esprit}. The resulting estimate is then used to extrapolate the process across time.  Such a prediction with unknown autocorrelation also arises in the problem of estimating the \textit{downlink} (DL) channel in a \textit{Frequency Division Duplex} (FDD) system from the \textit{uplink} (UL) observations. In such a system, the UL and DL transmissions are performed in different frequency bands, and the \textit{Base Station} (BS) needs to send pilot signals to the users in the DL and receive their feedback via UL to acquire the channel state information. To reduce the resulting feedback overhead, in \cite{vasisht2016eliminating} the authors proposed an estimator for the delay profile of the channel using the observations over UL, which they exploited to extrapolate the channel state for DL. As in \cite{shirani2010mimo}, they assume a discrete model (for delay profile), thus, overlooking the possibility of continuous components, which occur in practical scenarios.

In this paper, we generalize the idea proposed by \cite{shirani2010mimo,vasisht2016eliminating}  and design a blind predictor that does not require the knowledge of autocorrelation or PSD of the process. To do so, we use the random spectral representation of a stationary Gaussian process whose underlying structure depends on the PSD. Using this representation, we obtain a decomposition of the process $\bbh$ into a \textit{discrete} and a \textit{continuous} part denoted by $\bbh^d$ and $\bbh^c$, which resembles the well-known Wold's decomposition for stationary processes \cite{wold1939study}. We also show that in contrast to $\bbh^d$, which can be well predicted from its past samples, prediction of $\bbh^c$ typically results in a large error. Motivated by the theory developed in \cite{chandrasekaran2012convex, bhaskar2013atomic, tang2013compressed,candes2014towards}, we use an atomic-norm minimization approach to quantize the random spectrum of process $\bbh$, and use the resulting quantization to perform a parametric extrapolation of the process. Compared to the previous works, where atomic-norm minimization is exploited for denoising a sparse signal, in our case the underlying spectrum is not sparse per se and the atomic-norm minimization is used merely for quantization rather than denoising. 

%

The rest of the paper is organized as follows. In sections \ref{sec:preliminaries} and \ref{sec:problem_statement}, we introduce the necessary  tools from probability and stochastic processes and state the extrapolation problem to be solved. Section \ref{sec:proposed_alg} describes our spectrum quantization algorithm based on the current results on atomic-norm minimization. Finally, we provide simulation results to evaluate the performance of our proposed algorithm in section \ref{sec:simulation_results}.

\section{Preliminaries}\label{sec:preliminaries}
Earlier we introduced $\bbh$ as a stationary discrete-time circularly symmetric Gaussian process with a zero mean and an autocorrelation function $r_{h}(m)=\bE[h_nh^*_{n-m}]$. Without any loss of generality, we assume that $\bbh$ has a normalized power, i.e., $r_{h}(0) = \bE[|h_n|^2]=1$. From the spectral theory of stationary processes \cite{grimmett2001probability} and, in particular, the Wiener-Khinchin theorem, we have that 
\begin{equation}\label{WK1}
r_{h}(m) = \int_{-1/2}^{1/2} \mathrm{e}^{\mathrm{j}2\pi m \xi} F(\mathrm{d}\xi)
\end{equation} 
for some right-continuous and non-decreasing function $F: I \to [0,1]$, with $I:=[-\frac{1}{2}, \frac{1}{2}]$, $F(-\frac{1}{2})=0$ and $F(\frac{1}{2})=1$, which is known as the power spectral distribution of $\bbh$ and assigns the positive measure $\mu_F(a,b]=F(b)-F(a)$ to any half-open interval $(a,b]$ for $a<b$. When $F$ is dominated by the Lebesque measure (length) over $I$, i.e, $\mu_F(A)=0$ for any measurable subset $A$ of $I$ of  Lebesgue measure zero, then the derivative $\frac{\mathrm{d}F(\xi)}{\mathrm{d}\xi}$ exists almost everywhere \cite{rudin1987real}, and it is easier to work with the power spectral density (PSD) (also known as the spectral density function) defined as
$S_h(\xi)=\frac{\mathrm{d}F(\xi)}{\mathrm{d}\xi}$, for which the Wiener-Khinchin theorem takes the more familiar form given by
\begin{equation}\label{WK2}
r_{h}(m) = \int_{-1/2}^{1/2} \mathrm{e}^{\mathrm{j}2\pi m \xi} S_h(\xi) \mathrm{d}\xi.
\end{equation}
The spectral representations in \eqref{WK1} and \eqref{WK2} relate two deterministic functions, i.e., the autocorrelation function and the power spectral density function associated to the Gaussian process $\bbh$. In this paper, we are additionally interested in a random spectral representation of $\bbh$. Such a representation is itself a circularly symmetric Gaussian process $\bH:=\{H(\xi): \xi \in I \}$ parametrized with $I$. It has independent increments, i.e., for any $v,u,v',u'\in I$ with $v < u < v' < u'$,
\begin{equation}
\bE[ (H(u) - H(v)) (H(u') - H(v'))^* ] = 0.
\end{equation}
Moreover, for such a $v,u$, it satisfies
\begin{align}\label{eq:F_H_relation}
\bE[|H(u) - H(v)|^2]=\mu_F(v,u]=F(u)-F(v), 
\end{align}
which implies that the variance of the increment $H(u) - H(v)$ in $\bH$ is given by the power of the process $\bbh$ in the spectral interval $(v,u]$. Also, $\bbh$ and $\bH$ are related via the stochastic integral 
\begin{equation}\label{stochastic_integral}
h_n = \int_{-1/2}^{1/2} \mathrm{e}^{\mathrm{j}2\pi \xi n}H(\mathrm{d}\xi) 
\end{equation}
which resembles the Fourier transform relation between the deterministic functions $r_h$ and $F$ as in \eqref{WK1} (please refer to \cite{grimmett2001probability} for a more rigorous definition of \eqref{stochastic_integral}).
%
%

The random spectral  representation in \eqref{stochastic_integral} has an important implication, that is, the process $\bbh$ can be represented as the superposition of discrete-time complex exponentials of the type $\psi_n(\xi)=\mathrm{e}^{\mathrm{j}2\pi \xi n} $ for frequencies $\xi \in I$ with circularly symmetric Gaussian coefficients given by the increments $H(\mathrm{d}\xi)$ of the spectral process around every frequency $\xi$. In particular, according to (\ref{eq:F_H_relation}), if $F(\xi)$ is continuous at $\xi$ then the coefficient $H(\mathrm{d}\xi)$ is a random variable with infinitesimal variance. In contrast, if $F(\xi)$ has a jump at $\xi$ ($F(\xi^+)\neq F(\xi^-)$), then the corresponding complex exponential $\psi_n(\xi)=\mathrm{e}^{\mathrm{j}2\pi \xi n}$ has a coefficient of non-infinitesimal variance given by the height of the jump of the spectral distribution at $\xi$, i.e. $F(\xi^+)- F(\xi^-)$. A convenient way to represent the jumps in  $F(\xi)$ is via using Dirac deltas in the PSD $S_h(\xi)$, where the variance of $H(\mathrm{d}\xi)$ at these jump points is given by the coefficient of the delta. Notice that, from Lebesgue decomposition, any \textit{Cumulative Distribution Function} (CDF) over the real-line and in particular the interval $I$ can be written as a convex combination of a discrete part, a continuous part, and a singular (fractal) part as $F(\xi) = \alpha_d F_d(\xi)  +\alpha_c F_c(\xi)+ \alpha_s F_s(\xi)$ where $\alpha_d,\alpha_c,\alpha_s\geq 0$ with $\alpha_d+\alpha_c+\alpha_s=1$. For simplicity, we will focus on the case where $\alpha_s=0$ and $F(\xi) = \alpha F_d(\xi)  +(1-\alpha) F_c(\xi)$ for some $\alpha \in [0,1]$. Note that $F_c(\xi)$ is a continuous CDF without any jumps, and $F_d(\xi)$ is a discrete CDF, i.e., it is right-continuous and piecewise constant with jumps in $I$.

Using the conventions above, we can decompose the stationary process $\bbh$ into a \textit{discrete} part and a \textit{continuous} part. The discrete part is equal to the sum of the complex exponentials $\psi_n(\xi)$ at points $\{\xi_i\}_{i=1}^k$ corresponding to the jump points of $F_d(\xi)$ given by
\begin{equation}
h_n^{d} = \int_{-1/2}^{1/2} \mathrm{e}^{\mathrm{j}2\pi \xi n} H^{d}( \mathrm{d} \xi)=\sum_{i=1}^k X_i \mathrm{e}^{\mathrm{j}2\pi \xi_i n}, 
\end{equation} 
where the variables $\{X_i\}_{i=1}^k$ are independent zero-mean circularly symmetric Gaussian variables of variance $\alpha (F_d(\xi_i^+) - F_d(\xi_i^-) )$, and where $H^{d}(\xi)=\sum_{i: \xi_i\leq \xi} X_i$, for $\xi \in I$, is a spectral Gaussian process with jumps of size $X_i$ at location $\xi_i$ generated by $\alpha F_d(\xi)$ according to \eqref{eq:F_H_relation}.
Similarly, the continuous part is given by 
\begin{equation}
h_n^{c}  = \int_{-1/2}^{1/2} \mathrm{e}^{\mathrm{j}2\pi \xi n} H^{c}( \mathrm{d} \xi), 
\end{equation}     
where $H^{c}(\xi)$ corresponds to a random spectral process obtained from $\bH$ after removing the discrete part associated with $\{X_i, \xi_i\}_{i=1}^k$ and is generated by the continuous part $(1-\alpha) F_c(\xi)$ according to \eqref{eq:F_H_relation}. We will use these decompositions in the next section to formulate the prediction problem of the process $\bbh$ more clearly. 

\section{Statement of the Problem}\label{sec:problem_statement}
We begin this section with an example. Let $\bbh$ be a stationary process with a purely discrete spectral distribution function $F(\xi)$, that is, $F(\xi)=F_d(\xi)$. From the observations we made in section \ref{sec:preliminaries},  we can write $h_n=\sum_{i=1}^k X_i \mathrm{e}^{\mathrm{j} 2\pi n \xi_i}$, where $\{X_i\}_{i=1}^k$ are independent zero-mean circularly symmetric Gaussian random variables with variances $\{\sigma_i^2\}_{i=1}^k$, where $\{\xi_i\}_{i=1}^k$ denote the jump locations in $F_d(\xi)$, and where $\sigma_i^2=F_d(\xi_i^+)-F_d(\xi_i^-)$. Let $N\in \bZ_+$ be the length of the observation window and let us define the vector-valued function $\bfa: I\to \bC^N$ by $\bfa(\xi)=(1, \mathrm{e}^{\mathrm{j} 2\pi \xi}, \dots, \mathrm{e}^{\mathrm{j} 2\pi (N-1)\xi})^\transp$. 
Also, let $\bfh=(h_0, \dots, h_{N-1})^\transp$ be the $N$-dim vector consisting of the first $N$ components of $\bbh$ as before. We can write $\bfh=\sum_{i=1}^k X_i \bfa(\xi_i)$, where it is seen that, for the discrete spectrum $F_d(\xi)$, the observation vector $\bfh$ resembles the signal received in a \textit{uniform linear array} with $N$ elements from $k$ Gaussian sources with amplitudes  and DoAs $\{X_i, \xi_i\}_{i=1}^k$. Thus, provided that $k\ll N$ and the sources are sufficiently separable in $\xi$, the random spectral process can be well estimated from $\bfh$ even in presence of noise. The resulting estimate $\widehat{H}(\xi)$ can be used to obtain an estimate of any other sample of the process $\bbh$ as $\widehat{h}_n=\int \mathrm{e}^{\mathrm{j} 2 \pi \xi n} \widehat{H}(d\xi)$, especially those samples $\bfg=(h_N, \dots, h_{2N-1})^\transp$ belonging to our desired prediction window $\Omega=\{N,N+1, \dots, 2N-1\}$. Thus, it seems that under some regularity conditions on $F_d(\xi)$, the prediction over $\Omega$ is feasible to do. 

Now assume that $F(\xi)=F_c(\xi)$ consists of only a continuous part. For simplicity of illustration, let us take $F_c(\xi)=\xi+\frac{1}{2}$. We can check that $F_c$ induces a uniform measure over $I$. In particular, for such a uniform measure 
\begin{align}
r_h(m)=\left \{ \begin{array}{ll} 1 & m=0,\\ 0 & \text{otherwise,} \end{array} \right.
\end{align}
which implies that $\bbh$ is a white noise consisting of i.i.d. Gaussian samples. As a result, the samples $\bfg$ inside the window $\Omega$ can not be estimated from the observed samples. This turns out to be true for continuous distributions $F_c(\xi)$ other than the uniform one. 

In practice, $F(\xi)$ consists of both continuous and discrete parts. However, in terms of prediction over $\Omega$, we can only hope to predict the discrete part $\bfg^d$ of $\bfg=\bfg^d + \bfg^c$ over the window $\Omega$ from the discrete part $\bfh^d$ of the observation $\bfh=\bfh^d + \bfh^c$. In fact, we can apply the estimation technique we mentioned for the discrete part $\bfh^d$ to estimate $\widehat{H}_d(\xi)$, which we can exploit to predict $\bfg^d$. This results in a prediction error comparable to that of the MMSE predictor, especially, for a long-term prediction.
However, an obstacle to do this is that we have access only to a mixture of the discrete part $\bfh^d$ and the continuous part $\bfh^c$ through $\bfh=\bfh^d+ \bfh^c$ rather than $\bfh^d$ itself. In particular we have the following result.

\begin{theorem}\label{thm_1}
Let $\bfh=(h_0,\dots, h_{N-1})^\transp$ be a vector consisting of $N$ samples of a zero-mean circularly symmetric stationary Gaussian process $\bbh$ with the power spectral distribution $F(\xi)=P_c F_c(\xi)+(1-P_c) F_d(\xi)$ with $0\le P_c\le 1$. Let  $\text{MMSE}(T)$ denote the minimum mean squared error for estimating sample $h_{N+T}$ of process $\bbh$ given the observation vector $\bfh$. Then
$\liminf_{T\rightarrow\infty} \text{MMSE}(T) \ge P_c$. \hfill $\square$
\end{theorem}
\begin{proof}
The proof is provided in Appendix I.
\end{proof}

In the next section, we propose an algorithm that applies the atomic-norm minimization to the full observation $\bfh$ to obtain a quantized estimate $\widehat{H}(\xi)$ of the random spectral process $H(\xi)=H^d(\xi)+H^c(\xi)$. The main idea is that under suitable regularity conditions on the discrete part $F_d(\xi)$, we can decompose the resulting estimate as $\widehat{H}(\xi)= \widehat{H}^d(\xi) + \widehat{H}^c(\xi)$, where $\widehat{H}^d(\xi)$ provides a quite precise estimate of the true process $H^d(\xi)$, which we can exploit to predict the discrete part $\bfg^d$. In contrast, $\widehat{H}^c(\xi)$ is only a discrete approximation of $H^c(\xi)$. In particular, although $\widehat{H}^c(\xi)$ provides a good approximation of $\bfh^c$ over the observation window, i.e., 
\begin{align}
h^c_n = \int \mathrm{e}^{\mathrm{j} 2\pi \xi n} H^c(d\xi)\approx \int \mathrm{e}^{\mathrm{j} 2\pi \xi n} \widehat{H}^c(d\xi), n \in [N],
\end{align}
from our earlier explanation, we expect that the long-term prediction of $\bfg^c$ obtained via $\widehat{H}^c(\xi)$ be uncorrelated with $\bfg^c$. Overall, using the proposed algorithm we hope to predict $\bfg$ from observations $\bfh$ with an error that grows by twice the variance of the continuous part $\bfg^c$, thus, twice the best prediction error achieved by the MMSE predictor.

\section{Proposed algorithm}\label{sec:proposed_alg}
As motivated in section \ref{sec:problem_statement}, our goal is to quantize the spectrum of the process $\bbh$. To do so, we find among all the complex measures that fit our observation vector $\bfh$, the one with lowest total-variation norm. More precisely, we consider
\begin{equation}\label{eq:tv_min}
\begin{aligned}
\mu^\ast = &~ \underset{\mu}{\arg\min}
& & \Vert \mu\Vert_{\text{TV}}  ~~\text{subject to} ~~ \Vert \mathcal{F}_N\mu - \bfy\Vert_2 \le \epsilon,
\end{aligned}
\end{equation}
where $\mathcal{F}_N$ is a linear map returning the first $N$ frequency coefficients of the measure $\mu(\mathrm{d}\xi)$, i.e., $[\mathcal{F}_N\mu]_n = \int_{-1/2}^{1/2}\mathrm{e}^{\mathrm{j}2\pi n \xi}\mu(\mathrm{d}\xi),~n\in [N]$, and $\epsilon$ is an estimate of the $\ell_2$-norm of the additive Gaussian noise vector.
Optimization~(\ref{eq:tv_min}) returns a discrete measure of the form
\begin{equation}\label{eq:mu_tv}
\mu^\ast=\sum_{i=1}^l c_i ^\ast\delta_{\xi_i^\ast},
\end{equation}
where $\delta_{\xi_i}^\ast$ denotes a delta measure at $\xi_i^\ast \in I$. Intuitively speaking, in the noiseless case, minimizing the total-variation norm returns a consistent measure, i.e., we expect that $\xi_i^\ast$s in (\ref{eq:mu_tv}) belong to the support of the PSD.  
Considering the complex measure $\mu^\ast$ in (\ref{eq:mu_tv}), we have that $\mathcal{F}_N\mu^\ast=\sum_{i}c_i^\ast \bfa(\xi_i^\ast)$ and total-variation norm of this measure is simply given by $\sum_{i}|c_i^\ast|$. This implies that (\ref{eq:tv_min}) can be equivalently written as 
\begin{equation}\label{eq:x_atomic_min}
\begin{aligned}
\bfx^\ast = & ~ \underset{\bfx}{\arg\min}
& & \Vert \bfx\Vert_{\mathcal{D}} ~~\text{\text{s.t.}} ~~ \Vert \bfx - \bfy\Vert_2 \le \epsilon, \\
\end{aligned}
\end{equation}
where $\Vert \bfx\Vert_{\mathcal{D}}$ denotes the atomic norm of $\bfx$ over the dictionary $\mathcal{D}=\Big\{\bfa(\xi)\in\mathbb{C}^N, \xi\in I\Big\}$ defined by 
 \begin{equation}\label{eq:Atomic_Norm}
 \Vert \bfx\Vert_{\mathcal{D}}=\inf \Big\{\underset{i}{\sum}|c_i| : \exists\, \xi_i \text{ s.t. } \bfx=\underset{i}{\sum}c_i\bfa(\xi_i) \Big\}
 \end{equation}
 and where the optimal solution $\bfx^\ast$ of \eqref{eq:x_atomic_min} is given by $\bfx^\ast = \mathcal{F}_N\mu^\ast$ in terms of the optimal solution $\mu^\ast$ of \eqref{eq:tv_min}.
From the solution of (\ref{eq:x_atomic_min}) one can determine the frequencies $\xi_i^\ast$ and their corresponding coefficients $c_i^\ast$, which immediately specifies the measure $\mu^\ast$, showing that (\ref{eq:x_atomic_min}) and (\ref{eq:tv_min}) are equivalent.
\begin{figure}[t]
	\centering
	\includegraphics[ width=0.45\textwidth,height=0.3\textwidth]{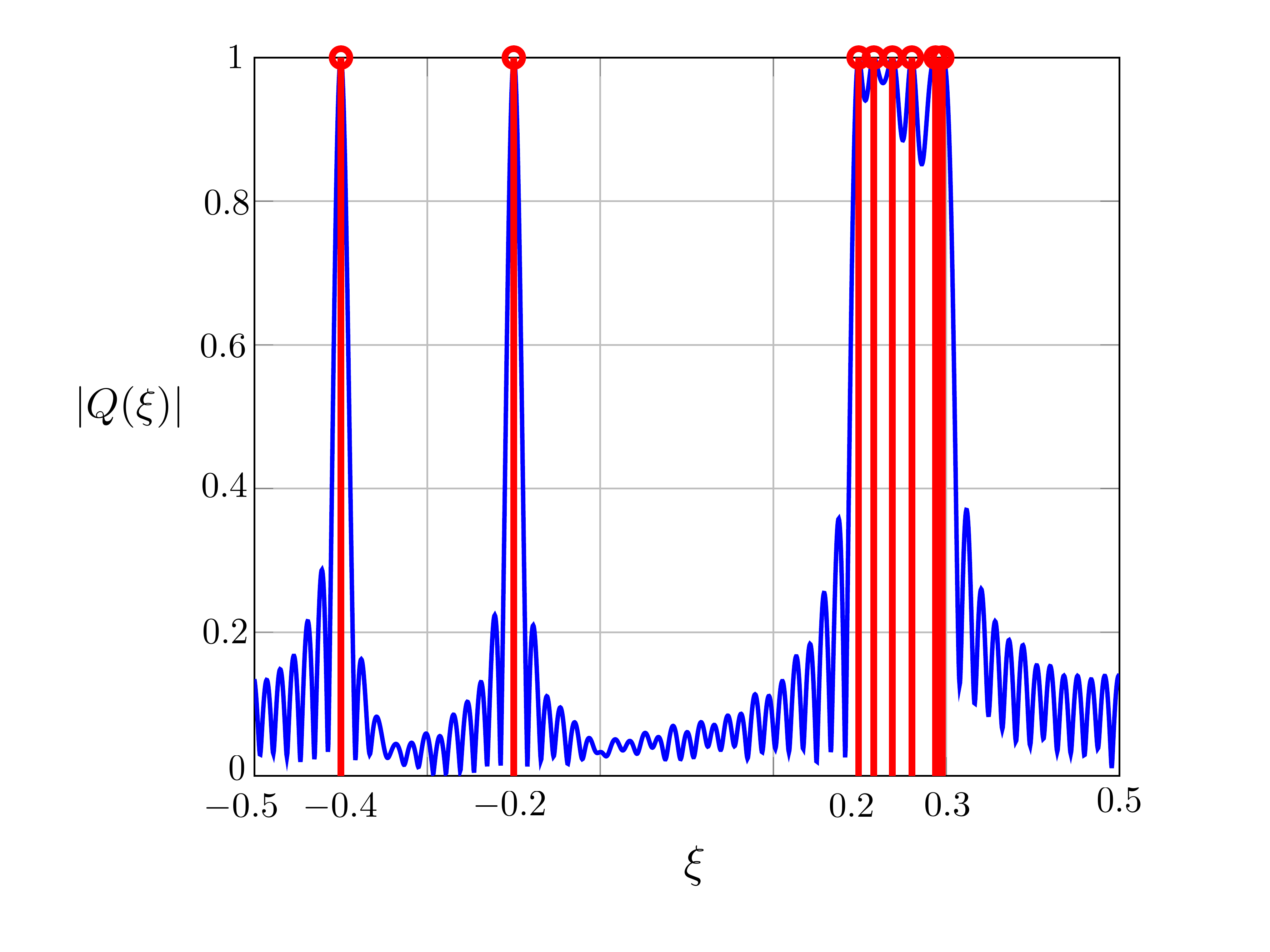}
	\caption{Dual polynomial $Q(\xi)$ as a function of $\xi$. The blue curve shows the dual polynomial $Q(\xi)$ and the red spikes show the quantization points $\{\xi_i\}_{i=1}^l$.}
	\label{fig:spectrum_quantize}
\end{figure} 

In particular, solving (\ref{eq:x_atomic_min}) serves two purposes: first, it reduces the effect of white additive noise and second, it forces the estimated vector to have a sparse representation in $\mathcal{D}$, which implies a sparse quantization of the PSD. Fortunately, compared to (\ref{eq:tv_min}) which requires an infinite-dimensional optimization over the space of complex measures, \eqref{eq:x_atomic_min} has an equivalent SDP form as 
\begin{equation}\label{eq:atomic_sdp}
\begin{aligned}
\{\bfx^\ast,\bfu^\ast,\lambda^\ast\} =&~ \underset{\bfx,\bfu,\lambda}{\arg\min}
 ~~ \frac{1}{2N}\text{trace}(\text{Toep}(\bfu))+\frac{1}{2}\lambda \\
 &~\text{s.t.}
~~~~~~
\begin{bmatrix}
\text{Toep}(\bfu) & \bfx\\
\bfx^\herm & \lambda
\end{bmatrix}\succeq 0,~~\Vert \bfx - \bfy\Vert_2 \le \epsilon,
\end{aligned}
\end{equation}
where $\text{Toep}(\bfu)$ denotes a Toeplitz hermitian matrix with $\bfu$ as its first column. Once this optimization problem is solved, one can specify the active atoms in $\mathcal{D}$, or equivalently the discrete measure $\mu^\ast$, by solving the dual problem to (\ref{eq:x_atomic_min}), which can also be represented as the following SDP \cite{tang2013compressed}:
\begin{equation}\label{dual_problem_1}
\begin{aligned}
\{\bfq^\ast,\bfH^\ast\}& =~ \underset{\bfq,\bfH}{\arg\max}
~~~~ \operatorname{Re}(\bfq^\herm \bfx^\ast)\\
& ~~~~~\text{s.t.} ~~ 
\begin{bmatrix}
\bfH & -\bfq\\
-\bfq^\herm & 1
\end{bmatrix}\succeq 0, \ \ \bfH=\bfH^\herm \\
&  ~~\sum_{k=1}^{N-j}H_{k,k+j}= 
\left\{
\begin{array}{ll}
1,j=0,\\
0, j=1,2,\ldots,N-1,
\end{array}
\right.\\
\end{aligned}
\end{equation}
Using the optimal solution $\bfq^\ast$ of the dual problem, we can construct the dual polynomial $Q(\xi)=\inp{\bfq^\ast}{\bfa(\xi)}=\sum_{n=0}^{N-1}q_n^\ast \mathrm{e}^{-\mathrm{j}2\pi n\xi}$ which satisfies the following \cite{tang2013compressed}
\begin{equation}\label{dual_poly_properties}
\begin{split}
& Q(\xi_i^\ast) = \text{sign}(c_i^\ast), ~~i=1,2,\ldots,l\\
& |	Q(\xi)|<1, ~~~~~~~~~~\xi\neq\xi_i^\ast ~~ \forall i
\end{split}
\end{equation}
where $\{c_i^\ast\}_{i=1}^l$ are the coefficients of the optimal complex measure $\mu^\ast$ in \eqref{eq:mu_tv} and where $\text{sign}(c_i^\ast)=\frac{c_i^\ast}{|c_i^\ast|}$. Therefore, the support $\{\xi_i^\ast\}_{i=1}^l$ of $\mu^\ast$ can be obtained from solving the equation $|Q(\xi)|=1$. If there is no noise and $\bfh^{c}=\bf0$, solving $|Q(\xi)|=1$ results in accurate spectrum localization provided that $k\le \frac{N}{2}$ and $\xi_i^\ast$s are separated enough such that 
\begin{equation}
\Delta_{\xi}=\underset{i\neq j}{\min}|\xi_i^\ast-\xi_j^\ast| \ge \frac{1}{\left\lfloor(N-1)/4\right\rfloor}.
\end{equation}
\begin{figure}[t!]
	\centering
	\includegraphics[trim={6cm 10cm 6cm 4cm},clip, width=0.45\textwidth,height=0.3\textwidth]{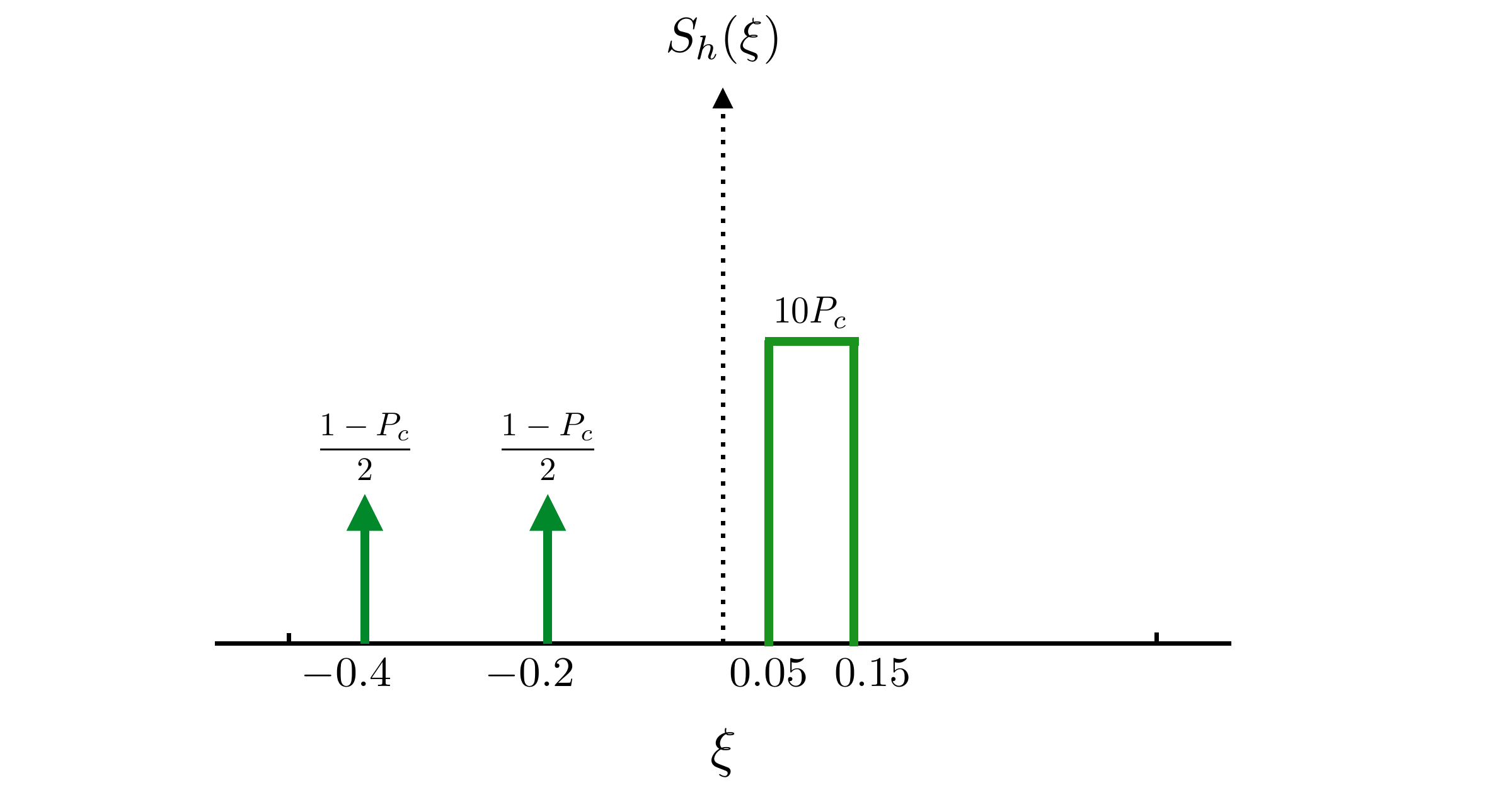}
	\caption{The PSD function used for simulation.}
	\label{fig:psd}
\end{figure}
However, this optimization applied to a mixture of discrete and continuous components, generates additional spectral elements to compensate for the the continuous component as well as noise. Fig. \ref{fig:spectrum_quantize} illustrates this fact more clearly. In this figure, we have plotted the corresponding dual polynomial in an experiment, where we observed a process consisting of a mixture of discrete and continuous components with $N=64$ and a PSD function which consists of a discrete part located at $\xi=-0.4$ and $\xi=-0.2$ and a continuous part uniformly distributed over $[0.2,0.3]$. Fig. \ref{fig:spectrum_quantize} illustrates how $\{\xi_i^\ast\}_{i=1}^l$ are identified via the dual polynomial, hence resulting in spectrum quantization. To obtain the coefficients corresponding to the set $\{\xi_i^\ast\}_{i=1}^{l}$, we form the matrix of active atoms as
\begin{equation}
\bfA=\Big[\bfa(\xi_1^\ast),\bfa(\xi_2^\ast),\ldots,\bfa(\xi_l^\ast)\Big].
\end{equation}
Knowing the active elements, the atomic-norm minimization in (\ref{eq:x_atomic_min})  can be written as an $\ell_1$-norm minimization
\begin{equation}\label{l_2_1_minimization}
\begin{aligned}
\bfc^\ast= &~ \underset{\bfc \in \mathbb{C}^{l}}{\text{arg min}}
&& \Vert \bfc\Vert_1  ~~~\text{s.t.}  ~~~\Vert \bfA\bfc - \bfy\Vert_2 \le \epsilon, \\
\end{aligned}
\end{equation}

where $\Vert \bfc\Vert_1 \coloneqq \sum_{i=1}^{l}|c_i|$. This optimization yields the corresponding coefficients $\bfc^\ast=(c_1^\ast,\ldots,c_l^\ast)^\transp$. Finally, we use the estimated parameters $\{c_i^\ast\}_{i=1}^l$ and $\{\xi_i^\ast\}_{i=1}^l$ to estimate $\bfg$ as
\begin{equation}
[\hat{\bfg}]_n=\sum_{i=1}^{l}c_i^\ast\mathrm{e}^{\mathrm{j}2\pi (N+n)\xi_i^\ast}.
\end{equation}
In the next section, we will show that the proposed method has a good performance in different scenarios.

\section{Simulation Results}\label{sec:simulation_results}
In this section, we provide simulation results to compare the performance of our proposed extrapolation method to that of the MMSE extrapolator. We focus on cases in which the process consists of a discrete and a continuous part. For simulations, we consider a noisy observation vector $\bfy$ of size $N=64$ with SNR = $20$ dB. Using this observation, we predict vector $\bfg$ consisting of the next $N=64$ samples of the process. Our performance metric in these simulations is the normalized mean squared error defined as $E = \frac{1}{N}\bE\{\Vert\hat{\bfg} - \bfg\Vert^2\}$. 
\subsection{Effect of the continuous component}
As before, we assume that $\int F(\mathrm{d}\xi)=1$ and denote the fraction of power in the continuous part by $P_c\in[0,1]$. The PSD function is chosen to have a discrete part located at $\xi=-0.4$ and $\xi=-0.2$ and a continuous part uniformly distributed over $[0.05,0.15]$, as plotted in Fig. \ref{fig:psd}. We increase $P_c$ from 0 to 0.5 and for each $P_c$ we generate $\bfy$, estimate $\bfg$ via both the MMSE estimator and our proposed estimator, and calculate the normalized estimation error $E$. We find an estimate of $E$ by averaging it over 1000 independent realizations of the process.  Fig. \ref{fig:sweep_alpha} illustrates $E$ as a function of $P_c$. It is seen that, the MMSE estimator has an error which is converging to $P_c$. The error of our proposed estimator is approximately twice the MMSE and, in fact, is quite low when $P_c$ is small. 

\begin{figure}[t]
	\centering
	\includegraphics[width=0.45\textwidth,height=0.3\textwidth]{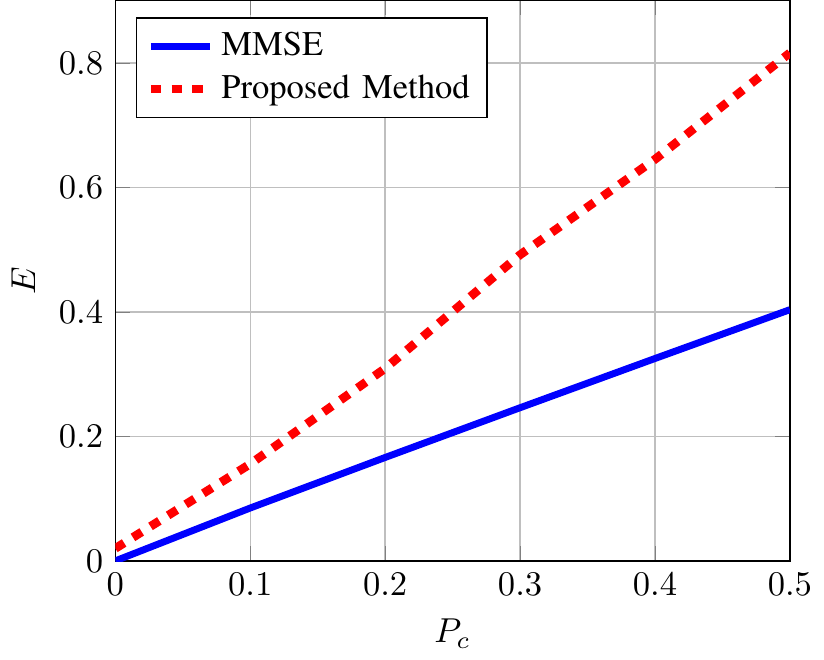}
	\caption{Normalized prediction error for MMSE predictor and our proposed estimator v.s. power of the continuous PSD $P_c$. }
	\label{fig:sweep_alpha}
\end{figure}

\subsection{Effect of increasing the number of jumps in $F_d(\xi)$}
To see how the number of jumps in the spectral distribution affects the performance of our estimator, we consider a PSD with a continuous part as before, i.e., uniformly distributed over $[0.05,0.15]$ and with a fixed power $P_c=0.3$. We add to this PSD $k$ Dirac deltas with equal  amplitudes $\frac{1-P_c}{k}=\frac{0.7}{k}$ and random frequencies $\{\xi_i\}_{i=1}^k\subset [0,1]$ with minimum separation larger than $\frac{1}{N}$. We perform the experiment for $k=1,\ldots,5$, and for each $k$ we repeat it for 1000 trials. Fig. \ref{fig:E_est} illustrates the result of this experiment. As the number of discrete elements grows, the error of the proposed method increases with a slight slope. This is because the power is distributed over a greater number of spikes and the noise effect appears relatively stronger. However, the error is still less than twice the MMSE. 

\section{Conclusion}
Using the framework of total-variation norm and atomic-norm minimization, we proposed a method to quantize the random spectrum of a stationary stochastic process. This quantization is then exploited to predict the process. We investigated the empirical performance of our proposed algorithm via numerical simulations. We illustrated that the prediction error is relatively low, and is roughly proportional to the MMSE up to a factor of 2 when there exists a continuous component in the PSD with non-negligible power.

\begin{figure}[t]
	\centering
	\includegraphics[trim=0cm 0cm 0 0.26cm, width=0.45\textwidth,height=0.3\textwidth]{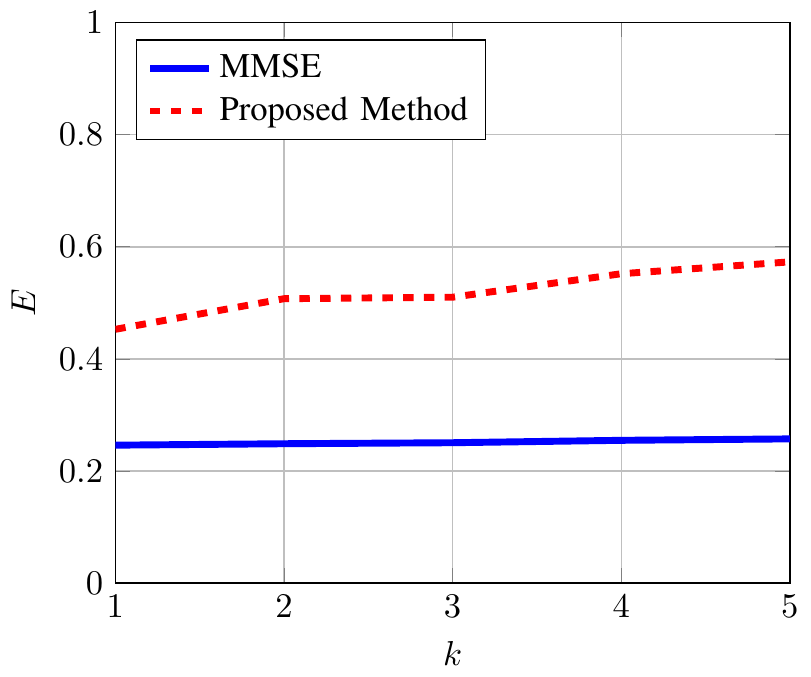}
	\caption{Normalized prediction error for MMSE predictor and our proposed estimator v.s. number of discrete PSD components $k$.  }
	\label{fig:E_est}
\end{figure}

\section{Appendices}
\subsection{Proof of Theorem \ref{thm_1}}\label{sec:appendix_1}
 The process $\bfh$ can be decomposed into independent discrete and continuous components, $\bfh^d$ and $\bfh^c$ with corresponding power spectral densities $(1-P_c)F_d(\xi)$ and $P_cF_c(\xi)$, respectively. Suppose there exists a genie-aided estimator which has access to the observation vectors $\bfh_c$ and $\bfh_d$ separately, as opposed to the conventional MMSE considered so far, which has access only to their sum, that is $\bfh = \bfh_c + \bfh_d$. We denote by $\text{MMSE}(T)$ the minimum mean squared estimation error for sample $h_{N+T}$.  We mention the mean squared error of the genie-aided estimator by $\text{MMSE}_{\text{genie}}(T)$. According to data processing inequality
\begin{equation}\label{mmse_ineq}
\text{MMSE}_{\text{genie}}(T) \le \text{MMSE}(T) ~~
\end{equation}
for all $T\ge 0$. Furthermore, the genie-aided estimator can be written as
\begin{equation}
\begin{aligned}
\hat{h}_{N+T}  & = \hat{h}_{N+T}^c + \hat{h}_{N+T}^d\\
\end{aligned}
\end{equation}
and since the continuous and discrete parts are independent, MMSE of the continuous part is only dependent on the continuous part and similarly MMSE of the discrete part is only dependent on the discrete part. As a result, the genie-aided estimator has the error 
\begin{equation}
\begin{aligned}
\text{MMSE}_{\text{genie}}(T) & = \bE[|h_{N+T}-\hat{h}_{N+T}|^2] \\
& =\bE[|h_{N+T}^c-\hat{h}_{N+T}^c|^2] + \bE[|h_{N+T}^d-\hat{h}_{N+T}^d|^2]
\end{aligned}
\end{equation}
In addition, it can be easily shown that
\begin{equation}\label{eq:est_err}
\begin{aligned}
\bE[|h_{N+T}^c-\hat{h}_{N+T}^c|^2] & = \bE[|h_{N+T}^c|^2] - \bE[|\hat{h}_{N+T}^c|^2] \\
& = P_c - \bE[|\hat{h}_{N+T}^c|^2]
\end{aligned}
\end{equation}
where
\begin{equation}
\bE[|\hat{h}_{N+T}^c|^2] = \Sigmam_{h_{N+T}^c\bfh^c} \Sigmam_{\bfh^c\bfh^c}^{-1} \Sigmam_{h_{N+T}^c\bfh^c}^\herm
\end{equation}
Note that $\Sigmam_{h_{N+T}^c\bfh^c} $ is an $N\times 1$ vector whose $l\textsuperscript{th}$ component $ l=0,\ldots, N-1$ is given by
\begin{equation}
\begin{aligned}
\Sigmam_{h_{N+T}^c\bfh^c} (l)& = P_c\int \textrm{e}^{\textrm{j}2\pi \xi (N+T-l)} F_c(\textrm{d}\xi)\\
& =P_c\int \textrm{e}^{\textrm{j}2\pi \xi (N+T-l)} S_c(\xi)\textrm{d}\xi ~~~~ l=0,\ldots, N-1
\end{aligned}
\end{equation}
where $S_c=\frac{dF}{d\xi}$ denotes the Radon-Nikodym derivative of the absolutely continuous measure $F_c$ with respect to the Lebesgue measure \cite{rudin1987real}. Since $\int |S_c(\xi)| d\xi=\int S_c(\xi) d\xi =P_c <\infty$, from Rimann-Lebesgue lemma \cite{bochner2016fourier}, it results that 
\begin{equation}
\lim_{T\rightarrow\infty} P_c\int \textrm{e}^{\textrm{j}2\pi \xi (N+T-l)} S_c(\xi)\textrm{d}\xi = 0.
\end{equation}
This implies that $\lim_{T\rightarrow\infty} \Sigmam_{h_{N+T}^c\bfh^c}=\bf0$ and as a result $\lim_{T\rightarrow\infty} \bE[|\hat{h}_{N+T}^c|^2]=0$. From  \eqref{eq:est_err} it results that $\lim_{T\rightarrow\infty} \bE[|h_{N+T}^c-\hat{h}_{N+T}^c|^2] = P_c - \lim_{T\rightarrow\infty} \bE[|\hat{h}_{N+T}^c|^2] = P_c$, which implies
\begin{equation}
\lim\inf_{T\rightarrow\infty} \text{MMSE}_{\text{genie}}(T) \ge P_c
\end{equation}
Plugging this inequality in \eqref{mmse_ineq} we obtain
\begin{equation}
\lim\inf_{T\rightarrow\infty} \text{MMSE}(T) \ge P_c.
\end{equation}
This completes the proof. $\blacksquare$

%
{\small
\bibliographystyle{IEEEtran}
\bibliography{references}
}

\end{document}